\documentclass[12pt]{amsart}

\pdfoutput = 1

\usepackage{macros}
\usepackage[normalem]{ulem}
\usepackage{float}

\usepackage{tikz-feynman}

\begin{document}
\onehalfspacing
	
\author{Fabian Hahner}

\email{fhahner@uw.edu}

\address{University of Washington, Department of Physics \\ 3910 15th Ave NE, Seattle, WA 98195,  U.S.A}

\title{On square-zero elements in the IIB supersymmetry algebra}

\begin{abstract}
We describe the variety of square-zero elements in the (2,0) super Poincar\'e algebra in ten dimensions, together with its orbit stratification induced by the action of the spin group and R-symmetry. This provides a classification of all possible twisting supercharges of type IIB supergravity in a flat background. Compared to other varieties of square-zero elements in super Poincar\'e algebras, we find that the orbit structure is more involved and exhibits several novel features. Among other things, we show that there are two distinct orbits of holomorphic twisting supercharges and a number of mixed topological-holomorphic twists.
\end{abstract}
	
	\maketitle
	\thispagestyle{empty}
	
	\setcounter{tocdepth}{1}
	
        \newpage
	
	\setlength{\parskip}{7pt}

\section{Introduction and Summary}
Given a super Lie algebra $\fg = (\fg_+ \oplus \Pi \fg_-,  [-,-])$, one can consider the space of odd square-zero elements
\begin{equation}
	Y = \{ Q \in \fg_- \: | \: [Q,Q] = 0 \}.
\end{equation}
This space, termed the nilpotence variety in~\cite{NV}, plays an important role in the study of supersymmetric field theories as well as representation theory~\cite{DS}. By definition, $Y$ is an affine variety cut out by homogeneous quadratic equations; often it is useful to consider its projective version. Further, it comes equipped with a group action by $G_+$, where $G_+$ denotes the simply connected Lie group integrating $\fg_+$.

The nilpotence varieties of super Poincar\'e algebras provide the natural moduli spaces of twists of supersymmetric field theories (and supergravity theories in flat backgrounds). Their description and thereby the corresponding classification of twisting supercharges was addressed for most examples of interest in the two seminal papers~\cite{ElliottSafronov, NV}. However, results on a few of these varieties are still missing. Crucially, this includes ten dimensional supersymmetry with 32 supercharges. In this note, we fill this gap for the ten-dimensional $\cN = (2,0)$ super Poincar\'e algebra.

Beyond classifying twisting supercharges, the nilpotence variety further plays a central role in the construction of supersymmetric field theories. In the pure spinor superfield formalism, supersymmetric multiplets and field theories are constructed from a (module over the) ring of functions on the variety of square-zero elements (see~\cite{Cederwall, Berkovits, perspectives, derivedps}). The algebraic geometry of the nilpotence variety is closely related to both superspace geometry and the representation theory of supersymmetric multiplets~\cite{CY2, 6dbundles, sconf}. Further, the approach via pure spinor superfields has turned out very fruitful for the computation of twists of supersymmetric field theories~\cite{spinortwist}. This is of particular interest in the context of twisted holography, where twisted supergravity theories are a crucial ingredient~\cite{CostelloLi, Costello:2018zrm, Costello:2020jbh}. Explicit computations of twists have been carried out with these methods, for instance, for eleven-dimensional supergravity~\cite{spinortwist,RSW,MaxTwist,CY2}.

For type IIB supersymmetry, the pure spinor superfield multiplet obtained from the functions on the nilpotence variety\footnote{The canonical multiplet in the terminology of~\cite{MSJI}.} is known to reproduce the field content of type IIB supergravity~\cite{NV}. In~\cite{CostelloLi}, it was conjectured that holomorphically twisted supergravity can be described via BCOV theory. This was verified at free level using pure spinor methods in~\cite{spinortwist}. 

Here, we find that there are two orbits of holomorphic twisting supercharges and that their stabilizers differ by a factor of $\CC^\times$. It would be interesting to see whether these lead to inequivalent twists of type IIB supergravity. Further, we find different orbits of mixed topological-holomorphic supercharges: a single orbit of twisting supercharges with four topological directions together with three orbits of twists with eight topological directions. The stabilizers in the latter cases contain copies of $\Spin(7)$, $\mathrm{SL}(4)$, and $\mathrm{Sp}(4)$ suggesting that global versions of these twists constrain the topological directions to be corresponding special holonomy manifolds.

An interesting feature of the IIB nilpotence variety is that the stratification by orbits of the spin group and R-symmetry is finer than the stratification by singularities. Typically, for nilpotence varieties of super Poincar\'e algebras with fewer supercharges both notions coincide~\cite{ElliottSafronov,NV}.

\subsection{Supersymmetry in ten dimensions and spinor orbits}
Let $(V, \langle-,-\rangle_V)$ be a complex ten-dimensional vector space equipped with a non-degenerate symmetric inner product. We denote the two 16-dimensional spinor representations of $\mathfrak{so}(V)$ by $S_\pm$. Recall that their symmetric square decomposes into $\mathfrak{so}(V)$-representations as $\Sym^2(S_{\pm}) = V \oplus Z_\pm$ where $Z_\pm$ are irreducible representations of dimension 126. We denote the projection on the vector representation by
\begin{equation}
\gamma : \Sym^2(S_\pm) \longrightarrow V.
\end{equation}
Further, we fix a two-dimensional vector space equipped with a non-degenerate symmetric inner product $W=(\CC^2, (-,-)_W)$. Then the complex $\cN=(2,0)$ supertranslation algebra in ten dimensions is the super Lie algebra with underlying $\ZZ_2$-graded vector space
\begin{equation}
\ft =  \Pi (S_+ \otimes W) \oplus V ,
\end{equation}
where the only non-vanishing brackets appear between two odd elements and is given by $\gamma \otimes (-,-)_W$. We extend this super Lie algebra to the ten-dimensional $\cN=(2,0)$ super Poincar\'e algebra by taking the semidirect product
\begin{equation}
\fp = \left( \mathfrak{so}(V) \oplus \mathfrak{o}(W) \right) \ltimes \ft.
\end{equation}
We are going to study the space of square-zero elements of this super Lie algebra together with its decomposition into orbits under the action of $G:= \Spin(V) \times \mathrm{O}(W) \times \CC^\times$, where $\CC^\times$ acts via rescaling (i.e. $Q \mapsto \alpha Q$ with $\alpha \in \C^\times$).

\numpar[][Spinor orbits of $\Spin(10)$]
Recall that the Dirac spinor representations can be obtained by fixing a maximally isotropic subspace $L \subset V$ and setting $S = \wedge^\bullet L^\vee$.\footnote{Note that the choice of $L$ picks an orientation of $V$.} In our case---as in any even dimension---it decomposes into the two chiral pieces according to
\begin{equation}
S_+ = \wedge^{\mathrm{even}} L^\vee = \wedge^0 L^\vee\oplus \wedge^2 L^\vee \oplus \wedge^4 L^\vee  \qquad \text{and} \qquad S_- = \wedge^{\mathrm{odd}} L^\vee = \wedge^1 L^\vee \oplus \wedge^3 L^\vee \oplus \wedge^5 L^\vee.
\end{equation}
Further, the vector representation can be viewed as $V \cong L \oplus L^\vee$ with both $L$ and $L^\vee$ being maximally isotropic and paired under $\langle -,-\rangle_V$.
We can expand a general element $\psi \in S_+$ as $\psi = \psi^{(0)} + \psi^{(2)} + \psi^{(4)}$ according to their form degree. Fixing a top form $\Omega \in \wedge^5 L^\vee$ with corresponding inverse polyvector $\Omega^{-1} \in \wedge^5 L$, the map $\gamma: \Sym^2(S_+) \longrightarrow V$ can be expressed as
\begin{equation}
\gamma(\psi, \psi) = \Omega^{-1}(\psi^{(0)} \wedge \psi^{(4)}) - \frac{1}{2} \Omega^{-1}(\psi^{(2)} \wedge \psi^{(2)}) + \psi^{(2)} \vee \Omega^{-1}(\psi^{(4)}).
\end{equation}
Note that the first two terms land in $L$ while the third term lands in $L^\vee$.

To every $\psi \in S_+$, we can associate an isotropic subspace given by its annihilator under Clifford multiplication
\begin{equation}
\mathrm{Ann}(\psi) = \{ v \in V \: | \: v \cdot \psi = 0 \} \subset V.
\end{equation}

\begin{dfn}
	A spinor $\psi \in S_+$ is called pure if $\mathrm{Ann}(\psi)$ is maximally isotropic. In this case, we will also write $\mathrm{Ann}(\psi)=: L_\psi$. 
\end{dfn}

In fact, any maximally isotropic subspace that is positive with respect to the orientation of $V$ defines a unique pure spinor up to scale so that the projective variety of pure spinors $\cP \subset \P(S_+)$ can be identified with the orthogonal Grassmanian $\cP = \mathrm{OG}(5,10)_+$. Crucially the pure spinors are precisely the solutions to the equation $\gamma(\psi,\psi)=0$; in other words $\mathrm{OG}(5,10)_+$ is the projective nilpotence variety for ten-dimensional $\cN=1$ supersymmetry. The affine variety of pure spinors is a cone over this orthogonal Grassmannian; its affine dimension is 11. Note that, if we identify $S_+ \cong \wedge^{\mathrm{even}} L_\psi$, we have (up to scale) $\psi = 1 \in \wedge^0 L_{\psi}^\vee$.

The orbits of the minimal spin representations under $\Spin(d)$ were classified by Igusa for $d\leq 12$~\cite{Igusa}. For $d=10$ one finds that $S_+-\{0\}$ decomposes into two orbits.
\begin{itemize}
	\item[---] \emph{The pure spinor orbit:} The pure spinors form a single orbit under $\Spin(10)$, the stabilizer of a point can be identified with $\mathrm{SL}(L) \ltimes \wedge^2 L$. Here, $N_{10} = \wedge^2 L$ is viewed as a ten-dimensional nilpotent group acting on $S_+$ via contraction, i.e. a group element $\alpha \in N_{10}$ acts via $\psi \mapsto \psi + \alpha \vee \psi$.
	\item[---] \emph{The generic (open) orbit:} This orbit consists of all elements with $\gamma(\psi ,\psi) \neq 0$. We will sometimes refer to these as impure spinors. These form an open set in $S_+$ on which $\Spin(10)$ acts transitively. The affine dimension of the orbit therefore is 16. The stabilizer of a point is isomorphic to $\Spin(7) \ltimes \CC^8$.
\end{itemize}

\numpar[][R-symmetry orbits]
The R-symmetry space $W= (\CC^2 , (-,-)_W)$ decomposes under $\mathrm{O}(W) \times \CC^\times$ into two different orbits, namely the orbit of isotropic vectors $\{w \in W \: | \: (w,w)=0 \}$ and non-isotropic vectors $\{ w \in W \: | \: (w,w) \neq 0 \}$. Note that there are two linearly independent isotropic lines in $W$ (in a basis where $(-,-)_W$ is the unit matrix these are spanned by $(1,i)$ and $(1,-i)$) and these two lines are exchanged by the reflection in $\mathrm{O}(W)$.

\subsection{Summary of the results}
Viewing supercharges $Q \in S_+ \otimes W$ as maps $Q: W^\vee \longrightarrow S_+$ our orbit classification is based on two invariants: first, the rank of this map and, second, the intersection pattern of its projectivized image $\P(S_Q)$ with the pure spinor variety $\cP \subset \P(S_+)$. Here, the intersection can be a line, two points, one point, or empty. In the rank one case, twisting supercharges $Q= \psi \otimes w$ with $\psi$ a pure spinor are further distinguished depending on whether $w$ is isotropic or not. Moreover, in the rank two case, there are no square-zero supercharges so that the intersection of the pure spinor variety with $\P(S_Q)$ is empty. We summarize the results in the following theorem.

\begin{thm}
	The variety of square-zero supercharges for ten-dimensional $\cN=(2,0)$ supersymmetry decomposes into the following orbits.
	\emph{
	\begin{table}[H]
		\centering
		\begin{tabular}{|l|l|l|l|l|}
			\hline
			& & & \\[-0.3cm]
			\text{Orbit} & \text{Stabilizer} & \text{Background} & \text{Representative} \\ \hline \hline
			& & & \\[-0.3cm]
			\underline{Rank 1:} & & & \\
			& & & \\[-0.3cm]
			$\{p\}$, iso & $\mathrm{SL}(5) \ltimes N_{10} \times \CC^\times$ & $\CC^5$ & $1 \otimes (1,i)$ \\ 
			& & & \\[-0.3cm]
			$\{p\}$, non-iso & $\mathrm{SL}(5) \ltimes N_{10}$  & $\CC^5$ & $1 \otimes (1,0)$\\ 
			& & & \\[-0.3cm]
			$\emptyset$, iso & $\Spin(7) \ltimes N_8 \times \CC^\times$ & $\R^{8} \times \CC$ & $(1 + e_{2345}^\vee) \otimes (1,i)$ \\ 
			& & & \\[-0.3cm] \hline
			& & & \\[-0.3cm]
			\underline{Rank 2:} & & & \\
			& & & \\[-0.3cm]
			$\P(S_Q)$ & $(\mathrm{SL}(2) \times \mathrm{SL}(3)) \ltimes N_{15} \times \CC^\times$ & $\RR^4 \times \CC^3$ & $1 \otimes (1,0) + e^\vee_{45} \otimes (0,1)$\\
			& & & \\[-0.3cm]
			$\{p_1,p_2\}$ & $\mathrm{SL}(4) \ltimes N_8$ & $\RR^8 \times \CC$ & $1 \otimes (1,0) + e^\vee_{2345} \otimes (0,1)$\\
			& & & \\[-0.3cm]
			$\{p\}$ & $\mathrm{Sp}(4) \ltimes N_{13} \times \CC^\times$& $\RR^{8} \times \CC$ & $1 \otimes (1,0) + (e^\vee_{23} + e^\vee_{45}) \otimes (1,i)$\\
			& & &\\[-0.3cm] \hline
		\end{tabular}
\end{table}}
\end{thm}
Here, the first column labels orbit by rank and intersection $\P(S_Q) \cap \cP$; for the two holomorphic twists, there is a further distinction by R-symmetry data. The representatives are with respect to a basis $(e_1^\vee , \dots , e_5^\vee)$ of $L^\vee$ where we have identified $S_+ = \wedge^{\mathrm{even}}L^\vee$ and use the shorthand $e_{ij}^\vee = e_i^\vee \wedge e_j^\vee$ etc. Further, $N_d$ denotes a nilpotent group of dimension $d$.

The `background' column deserves a bit of explanation. Given a twisting supercharge, $\Im[Q,-] \subseteq V$ defines the subspace of translations that act trivially on the twisted theory. Since the dimension of $\Im[Q,-]$ is always at least half the dimension of $V$, we can decompose $V = \Im[Q,-] \oplus K^\vee$ such that $K^\vee$ is an isotropic subspace. As such, we can further decompose $V = W \oplus K \oplus K^\vee$ where $K$ and $K^\vee$ are paired under $\langle -,-\rangle_V$ and $W$ is a non-degenerate subspace given by $W = \Im[Q,-] \cap (K \oplus K^\vee)^\perp$. This is the decomposition into topological, anti-holomorphic, and holomorphic directions respectively. Indeed, after introducing the real form corresponding to Euclidean signature, the holomorphic and anti-holomorphic directions are exchanged under complex conjugation while the topological directions are stable under complex conjugation. Thus, choosing a twisting supercharge identifies the flat supergravity background $\R^{10}$ with $\R^{p} \times \C^q$ for some $(p,q)$ with $p + 2q = 10$.

In total, the orbit stratification is summarized in the following diagram. Here, we label orbits by their projective dimension together with the twisted background spacetime and the Levi of the stabilizer.
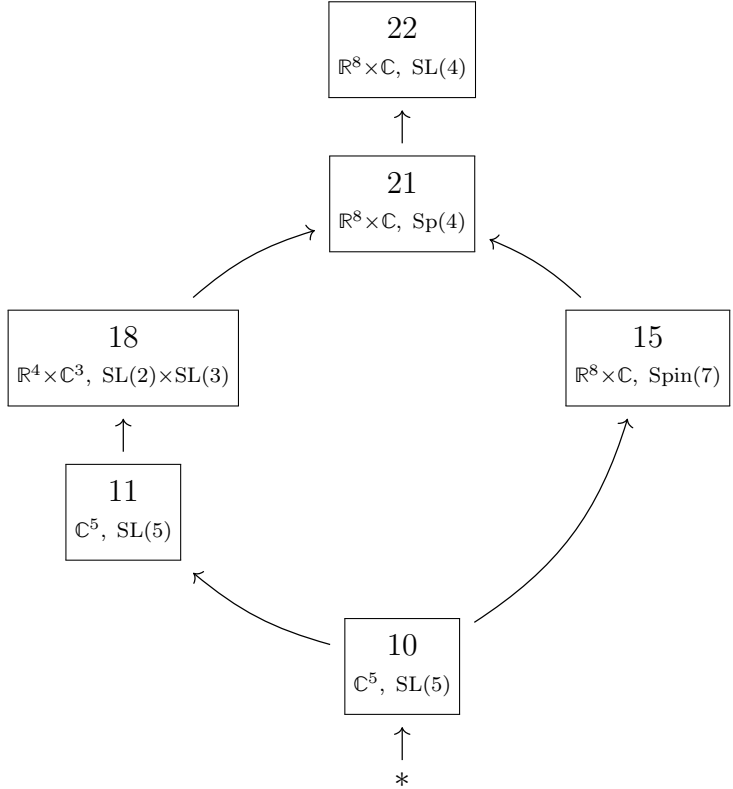
\[
\begin{tikzcd}[row sep=1.1em, column sep=2em, every node/.style={inner sep=3pt}]
& & & {\boxed{\begin{gathered}22\\[-0.7em] {\scriptstyle{\RR^8\times\CC,\ \mathrm{SL}(4)} }\end{gathered}}} & & & & \\
& & & {\boxed{\begin{gathered}21\\[-0.7em] {\scriptstyle \RR^8\times\CC,\ \mathrm{Sp}(4)}\end{gathered}}} 
\arrow[u] & & & & \\
& &  {\boxed{\begin{gathered}18\\[-0.7em] {\scriptstyle \RR^4\times\CC^3,\ \mathrm{SL}(2)\times\mathrm{SL}(3)}\end{gathered}}} 
\arrow[ur, bend left=12]  & & {\boxed{\begin{gathered}15\\[-0.7em] {\scriptstyle \RR^8\times\CC,\ \mathrm{Spin}(7)}\end{gathered}}} 
\arrow[ul, bend right=10] & & & \\
& & {\boxed{\begin{gathered}11\\[-0.7em] {\scriptstyle \CC^5,\ \mathrm{SL}(5)}\end{gathered}}} 
\arrow[u] & & & & & \\
& & & {\boxed{\begin{gathered}10\\[-0.7em] {\scriptstyle \CC^5,\ \mathrm{SL}(5)}\end{gathered}}} 
\arrow[uur, bend right=20] 
\arrow[ul, bend left=12] & & & & \\
& & & * \arrow[u] & & & &
\end{tikzcd}
\]

We provide a partial description of the geometry of these orbits. The 10-dimensional orbit can be identified with two copies of the pure spinor variety $\cP$, while the closure of the 11-dimensional orbit is $\cP \times \P^1$. The 15-dimensional stratum is two copies of $\P(S_+)$ and the 18-dimensional orbit forms a $\mathrm{PGL}(2)$-bundle over the orthogonal Grassmannian $\mathrm{OG}(3,10)$.

\subsection*{Acknowledgements}
I would like to thank Simone Noja, Natalie Paquette, and Ingmar Saberi for helpful discussions related to this project.

\section{Classification of orbits}
We will view supercharges as maps $Q: W^\vee \longrightarrow S_+$ and denote their image by $S_Q \subset S_+$. A first distinction between different supercharges is provided by the rank of the supercharge, $\rank(Q) = \dim(S_Q)$. In many of the examples in lower dimensions and with less supersymmetry, the rank is already a complete invariant of orbits in the nilpotence variety and thereby sufficient for the classification of twists (see~\cite{ElliottSafronov, NV}). As we will see now, this is not true for type IIB supersymmetry.

\subsection{Rank one orbits}
For elements $Q= \psi \otimes w$ of rank one, the square-zero condition is simply
\begin{equation}
	[Q,Q] = \gamma(\psi, \psi) (w,w) = 0.
\end{equation}
Thus, we find that $\psi \in S_+$ has to be a pure spinor or $w \in W$ has to be isotropic. Accordingly, we have found that the rank one square-zero elements decompose into three different orbits.

\numpar[][The pure-isotropic orbit]
First, we have the option where $\psi$ is a pure spinor, i.e. $\gamma(\psi, \psi)=0$ and $w$ is isotropic. It is easy to see that all such supercharges are related by $G = \Spin(V) \times \mathrm{O}(W) \times \CC^\times$. 

Further, we can describe this orbit explicitly: Let $\cC \cP^\times \subset S_+$ denote the affine cone over the pure spinor variety and $\cC_{iso}^\times \subset W$ the cone of isotropic lines in $W$, both with the origins removed. Then we can take a pair $(\psi, w) \in \cC \cP^\times \times \cC_{iso}^\times$ and form $\psi \otimes w$ to get an element of this orbit. This is not injective, instead we have to take the quotient by the $\CC^\times$-action that identifies $(\psi, w) \sim (\alpha \psi, \alpha^{-1} \psi)$. Together with the $\CC^\times$ action from the projectivization we identify the orbit as $(\cC \cP^\times \times \cC_{iso}^\times)/(\CC^\times \times \CC^\times) \cong \cP \times OG(1,2)$, i.e. two copies of the pure spinor variety which are exchanged under the reflection in $\mathrm{O}(2)$.

Since $\psi$ is a pure spinor $\Im([Q,-]) \subset V$ is a maximally isotropic subspace so that the associated twisting supercharges are holomorphic. The stabilizer of the line spanned by any given supercharge in this orbit is $\mathrm{SL}(5) \ltimes \wedge^2 L \times \CC^\times$.

\numpar[][The pure-nonisotropic orbit]
Second, we find a family of square-zero supercharges, where $\psi$ is a pure spinor, but $w$ lies in the nonisotropic orbit inside $W$. It is clear that the closure of this orbit encompasses the pure-isotropic orbit discussed above. These supercharges correspond to the holomorphic twist of type IIB supergravity as studied in~\cite{CostelloLi}. The stabilizer of a point in the projective orbit is $\mathrm{SL}(5) \ltimes \wedge^2 L$, i.e. it differs from the pure-isotropic case just by a factor of $\CC^\times$. The closure of this orbit can be identified with $\cP \times \P^1$.

\numpar[][The impure spinor orbit]
If $\psi$ is not a pure spinor, then the square-zero condition is equivalent to $w\in W$ being isotropic. Recall that all impure spinors in $S_+$ lie in the open orbit of $\Spin(V)$. Projectively, the two isotropic lines in $W$ are reduced to two points, so that we can identify this orbit as a product of the open orbit of impure spinors in $S_+$ with two points. Note that the two connected components are exchanged by $\mathrm{diag}(1,-1) \in \mathrm{O}(W)$. In particular, the projective dimension is 15 and the stabilizer of a point is $\Spin(7) \ltimes \CC^8 \times \CC^\times$.

A short calculation shows that there is a single translation surviving in twists of this type.
\begin{prop}
	For $Q$ in the impure spinor orbit, we have $\dim(\im([Q,-])) = 9$ so that the corresponding twist of type IIB supergravity in a flat background lives on $\RR^8 \times \CC$.
\end{prop}
\begin{proof}
	To see this, we pick an impure spinor $\psi \in S_+$ and compute $\Im(\gamma(\psi,-))$. For example, take
	\begin{equation}
		\psi = 1 + e^\vee_{2345}. 
	\end{equation}
	Then we have $\gamma(\psi , e^\vee_1 \wedge \dots \Hat{e}_j^\vee \dots \wedge e^\vee_5) = e_j$ for $j = 1, \dots ,5$ (where the hat signifies omission) and $\gamma(\psi, e^\vee_i \wedge e^\vee_1) = e^\vee_i$ for $i=2 \dots 5$ so that we find $\Im(\gamma(\psi,-)) = \mathrm{span}(e_1,\dots e_5, e^\vee_2, \dots , e^\vee_5)$. Since the dimension of  $\Im([Q,-])$ is invariant under $\Spin(V)$, it follows that this is true for any impure spinor. Writing $Q= \psi \otimes w$ with $w$ in the isotropic line spanned by $(1,i)$ the claim follows.
\end{proof}

\subsection{Rank two elements}
Let us now turn our attention towards supercharges where the image $S_Q$ is of dimension two. We can expand a generic rank two element as
\begin{equation}
	Q = \psi_1 \otimes w_1 + \psi_2 \otimes w_2 = \begin{pmatrix}
	\psi_1 & \psi_2
	\end{pmatrix} \otimes
	\begin{pmatrix}
	w_1 \\ w_2
	\end{pmatrix}.
\end{equation}
However, we note that such a decomposition is non-canonical. Indeed, for any $A \in \GL(2,\CC)$ we can act via
\begin{equation}
	\begin{pmatrix}
	\psi_1 & \psi_2
	\end{pmatrix} \mapsto
	\begin{pmatrix}
	\psi_1 & \psi_2
	\end{pmatrix}
	A^{-1} \qquad \text{and} \qquad 
	\begin{pmatrix}
	w_1 \\ w_2
	\end{pmatrix} \mapsto
	A \begin{pmatrix}
	w_1 \\ w_2
	\end{pmatrix}
\end{equation}
to change the decomposition but leaving $Q$ invariant. 

A more invariant way to approach the orbit structure of rank two supercharges is the following. Recall that $\Spin(V)$ decomposes $\P(S_+)$ into two orbits: the projective pure spinor variety $\cP$ and the generic open orbit. Since the space of pure spinors $\cP$ is cut out by a set of quadratic equations, the intersection pattern of the projective line $\P(S_Q)$ with $\cP$ distinguishes four different cases:
\begin{itemize}
	\item[1.] $\cP \cap \P(S_Q) = \P(S_Q)$: If the image of a given supercharge is contained in the pure spinor variety, this implies that any possible decomposition of the supercharge has $\psi_1$ and $\psi_2$ pure.
	\item[2.] $\cP \cap \P(S_Q) = \{p_1, p_2\}$: If the image intersects with the pure spinor variety in two points, this means that it is possible to choose a decomposition with $\psi_1$ and $\psi_2$ pure, but for a generic decomposition, this will not be true.
	\item[3.] $\cP \cap \P(S_Q) = \{p\}$: If the image intersects in a single point (i.e. the line $\P(S_Q)$ is tangent to $\cP$), then there exists a decomposition where $\psi_1$ is pure and $\psi_2$ lies in the generic orbit.
	\item[4.] $\cP \cap \P(S_Q) = \emptyset$: If the intersection with the image is empty, then, for any decomposition, $\psi_1$ and $\psi_2$ lie in the generic orbit.
\end{itemize}
We now investigate these cases individually.

\numpar[][$\cP \cap \P(S_Q) = \P(S_Q)$]
First, we note that supercharges of this type are automatically square-zero. Indeed pick any decomposition and expand:
\begin{equation}
[Q,Q] = \gamma(\psi_1,\psi_1) (w_1, w_1)_W + 2\gamma(\psi_1,\psi_2) (w_1, w_2)_W + \gamma(\psi_2,\psi_2) (w_2 ,w_2)_W =0.
\end{equation}
Since $\psi_1$ and $\psi_2$ are pure the first and the last term vanish. The middle term vanishes since any linear combination of them is also pure.

Thus, specifying a supercharge of this type means choosing an appropriate two-dimensional subspace $U \subset S_+$ together with an isomorphism $\varphi: W^\vee \longrightarrow U$. We now show that all such supercharges form a single orbit under $G$.
\begin{prop} \label{prop: orbit}
	$G$ acts transitively on pairs $(U, \varphi)$ where $U \subset S_+$ is a two-dimensional subspace so that $\cP \cap \P(S_Q) = \P(S_Q)$ and $\varphi: W^\vee \longrightarrow U$ is a linear isomorphism.
\end{prop}

\begin{proof}
	Let $(U,\varphi)$ and $(U', \varphi')$ be two such pairs. First, we notice that $\Spin(V)$ acts transitively on the set of lines in the pure spinor variety $\cP$. (See for example~\cite{Kuznetsov_2018, LiuManivel}; the Fano variety of lines in $\cP$ is again a homogeneous space, in fact it can be identified with $\mathrm{OG}(3,10)$.) This means we can identify, $U'$ with $U$ after transforming with a suitable element $g \in \Spin(V)$.
	
	It remains to relate two linear isomorphisms $\varphi, \varphi' : W^\vee \longrightarrow U$ by $\mathrm{Stab}(U) \times \mathrm{O}(W) \times \CC^\times$, where $\mathrm{Stab}(U) \subset \Spin(V)$ is the stabilizer of $U$ under the action of the spin-group. However, note that the stabilizer $\mathrm{Stab}(U)$ contains a factor of $\mathrm{GL}(2)$ so that we can always relate two such linear isomorphisms. To be completely explicit, we can choose a basis $(e_1,\dots, e_5, e_1^\vee, \dots , e_5^\vee)$ of $V = L \oplus L^\vee$ so that $\psi_1 = 1$ and $\psi_2 = e_4^\vee \wedge e_5^\vee$. Then $\mathfrak{so}(10)$ contains elements acting via $E = (e_4^\vee \wedge e_5^\vee) \wedge (-)$ and $F = (e_4 \wedge e_5) \vee (-)$. Clearly these leave $U = \mathrm{span}(\psi_1, \psi_2)$ invariant and form with their commutator $H:=[E,F]$ an $\mathrm{SL}(2)$-triple. Further, consider the $\GL(2)$ subgroup of $\Spin(10)$ that stabilizes $\mathrm{span}(e_4^\vee, e_5^\vee)$. This subgroup also stabilizes $U$ and elements $g \in \GL(2)$ act via
	\begin{equation}
	1 \mapsto 1 \qquad e_4^\vee \wedge e_5^\vee \mapsto \det(g) e_4^\vee \wedge e_5^\vee.
	\end{equation}
	Together, we obtain a copy of $\GL(2)$ as claimed.
\end{proof}

Let us now investigate these supercharges in terms of the pairs of pure spinors $(\psi_1, \psi_2)$ that define them in any decomposition. Let us recall the following facts on pure spinors.
\begin{lem}
	Let $(\psi_1, \psi_2)$ be a pair of pure spinors associated to maximally isotropic subspaces $L_{\psi_1} , L_{\psi_2} \subset V$. Let $r = \dim(L_{\psi_1} \cap L_{\psi_2})$ be the intersection dimension of a pair of pure spinors.
	\begin{itemize}
	\item[1.] $\psi_1$ and $\psi_2$ are of the same chirality (i.e. $\psi_1, \psi_2 \in S_+$) if and only if $r$ is odd.
	\item[2.] $a \psi_1 + b \psi_2$ is pure for all $a,b \in \CC^\times$ if and only if $r \in \{5,3\}$.
	\item[3.] Let $(\psi'_1 , \psi'_2)$ be another pair of pure spinors with intersection dimension $r'$. They are related to $(\psi_1, \psi_2)$ up to scale via the diagonal action of $\Spin(V)$ if and only if $r' =r$
	\end{itemize}
\end{lem}
Hence, we see that supercharges in this orbit precisely arise from pairs of pure spinors with intersection dimension three. (Note that for $r=5$, the $\psi_1$ and $\psi_2$ are already proportional so that we are back to the rank one case.)
\begin{proof}
	The first two items are standard, see e.g~\cite{CharltonThesis, Budinich:1989bg}
	
	For the third part, note that the intersection dimension is invariant under the action of $\Spin(V)$ which establishes that $r=r'$ is necessary for the two pairs to be related via the group action. To see that it is sufficient, we now construct an orthogonal transformation $f: V \longrightarrow V$ that maps the pair $(L_{\psi_1} , L_{\psi_2})$ to $(L_{\psi'_1},L_{\psi'_2})$. This can be arranged explicitly as follows.
	
	Let $K = L_{\psi_1} \cap L_{\psi_2}$. Since $K$ is isotropic, we can choose another isotropic subspace $K^\vee \subset V$ so that the bilinear form $\langle -,-\rangle_V$ on $V$ restricts to a perfect pairing $K \times K^\vee \longrightarrow \CC$. Further, we choose complementary subspaces $L_{\psi_1} = K \oplus A$ and $L_{\psi_2} = K \oplus B$. Then $\langle -,-\rangle_V$ also restricts to a perfect pairing $A \times B \longrightarrow \CC$.
	The fact that this pairing is perfect follows from the maximality of $L_{\psi_1}$ and $L_{\psi_2}$.
	We apply the same procedure to $K' = L_{\psi'_1} \cap L_{\psi'_2}$.
	
	In total, we obtain two decompositions of $V$
	\begin{equation}
	V = K \oplus K^\vee \oplus  A \oplus B \quad \text{and} \quad V = K' \oplus K'^\vee \oplus  A' \oplus B' ,
	\end{equation}
	together with bases
	\begin{equation}
	(u_i,u^\vee_i, a_j, b_j) \quad \text{and} \quad (u'_i,u'^\vee_i, a'_j, b'_j) \qquad i=1 \dots r, j = 1 \dots 5-r
	\end{equation}
	that both bring $\langle-,-\rangle_V$ to a standard block diagonal form.
	
	We define $f$ by mapping the unprimed basis to the primed basis. Then it is clear that $f$ interchanges the relevant subspaces and further preserves the bilinear form and is hence orthogonal.
\end{proof}

We see that these supercharges of this type give rise to what is often known as the $\mathrm{SL}(3)$-twist of type IIB supergravity.
\begin{lem} \label{lem: r= 3}
	Supercharges in this orbit have three surviving translations. The corresponding twist of type IIB supergravity in a flat background lives on $\CC^3 \times \RR^4$. The stabilizer of any such supercharge is $(\mathrm{SL(3)} \times \mathrm{SL}(2)) \ltimes N_{15} \times \CC^\times$.
\end{lem}
\begin{proof}
	Recall that $L_{\psi_i} = \Im(\gamma(\psi_i , -))$. Thus, we have $\Im([Q,-]) = L_{\psi_1} + L_{\psi_2} \subset V$ and in particular $\dim(\Im([Q,-])) = 7$. Denoting $K = L_{\psi_1} \cap L_{\psi_2}$, we can decompose $V=K^\vee \oplus K \oplus K^\perp /K$. This is the decomposition into holomorphic, antiholomorphic and topological directions. 
	
	In order to compute the stabilizer, we pick a representative of this orbit for example
	\begin{equation}
		Q = 1 \otimes (1,0) + e_4^\vee \wedge e_5^\vee \otimes (0,1).
	\end{equation}
	We decompose an element $g \in \so(10) \oplus \mathfrak{o}(2)$ according to $(X_+, X_-, A, t)$ where $X_+ \in \wedge^2 L^\vee$ acts via wedging, $X_- \in \wedge^2 L$ via contraction, and $A \in L \otimes L^\vee$ as a degree preserving endomorphism while $t$ is the parameter of the rotation in $\mathfrak{o}(2)$. The condition to leave $Q$ invariant, seperates into three conditions corresponding to the degree in $S_+ \cong \wedge^\bu L^\vee$.
	\begin{equation}
	\begin{matrix}
	\text{deg}=0: & \tr(A) = 0 & \text{and} & X_- \vee (e_4^\vee \wedge e_5^\vee) = t \\
	\text{deg}=2: & A(e_4^\vee \wedge e_5^\vee) = 0 & \text{and} & X_+ = t e^\vee_4 \wedge e^\vee_5 \\
	\text{deg}=4: & X_+ \wedge e_4^\vee \wedge e_5^\vee = 0
	\end{matrix}
	\end{equation}
	Thus, we immediately see that $X_+$ is completely fixed; the degree zero condition imposes a single equation on $X_-$ leaving nine generators preserving $Q$. The degree two condition introduces a block decomposition of $A$, where the diagonal blocks form the semisimple piece $\mathrm{SL}(3) \times \mathrm{SL}(2)$, the upper-triangular block has to vanish while the lower-triangular $3 \times 2$-block combines with the components of $X_-$ to form $N_{15}$. Together with the action of $\CC^\times \cong O(W)$, the claim follows.
\end{proof}

\numpar[][$\cP \cap \P(S_Q) = \{p_1, p_2\}$]
In this case, we can find a decomposition of $Q$ in terms of two pure spinors $(\psi_1, \psi_2)$, however, an arbitrary decomposition will not be of this form. This means that there is a basis of $S_Q$ given by pure spinors but linear combinations are generally not pure. This identifies $(\psi_1, \psi_2)$ as a pair of pure spinors with intersection dimension $r=1$. We note that such supercharges are not automatically square-zero. Indeed, the square-zero condition collapses to
\begin{equation}
[Q,Q] = 2 \gamma(\psi_1, \psi_2) (w_1 ,w_2)_W = 0.
\end{equation}
Since $\psi_1 + \psi_2$ is not pure, we have $\gamma(\psi_1, \psi_2) \neq 0$, which forces $(w_1, w_2)_W = 0$.

\begin{prop}
	$G$ acts transitively on square-zero supercharges of this type.
\end{prop}
\begin{proof}
	We show this by transforming an arbitrary supercharge $Q' = \psi'_1 \otimes w'_1 + \psi'_2 \otimes w'_2$ to a standard form $Q = \psi_1 \otimes (1,0) + \psi_2 \otimes (0,1)$ where $\psi_1 = 1 \in \wedge^0 L^\vee$ and $\psi_2 = e^\vee_{2345} \in \wedge^4 L^\vee$.

	First, we can choose a $g \in \Spin(V)$ to align $g(\psi'_1) = a \psi_1$ and $g(\psi'_2) = b \psi_2$ up to scales. 
	Recall that $w'_1$ and $w_{2}$ are linearly independent and orthogonal, so that they are in particular both non-isotropic. 
	Thus, we can use $\mathrm{O}(W)$ and $\CC^\times$ to achieve
	\begin{equation}
		Q' \mapsto \psi_1 \otimes (1,0) + t \psi_2 \otimes (0,1) \qquad \text{for some } t \in \CC^\times.
	\end{equation}
	Finally we note that $\Spin(10)$ contains $\GL(A)$ as a subgroup where $A = \mathrm{span}(e_2^\vee, \dots , e_5^\vee)$. Clearly any $g \in \GL(A)$ stabilizes $S_Q$ and acts via $1 \mapsto 1$ and $\psi_2 \mapsto \det(g) \psi_2$ so that we can use this action to set $t=1$.
\end{proof}

With the same reasoning as in the previous paragraph, we find that such twists have precisely one holomorphic direction.
\begin{lem}
	Supercharges of this type have one surviving translation. The corresponding flat background twist lives on $\CC \times \RR^8$. The stabilizer is $\mathrm{SL}(4) \ltimes N_8$.
\end{lem}

\begin{proof}
	The stabilizer can be computed from an explicit representative $Q = 1 \otimes (1,0) + e_{2345}^\vee \otimes (0,1)$. Using the same notation as in the proof of Lemma~\ref{lem: r= 3} we find the following conditions:
	\begin{equation}
	\begin{matrix}
	\text{deg}=0: & \tr(A) = 0 & \text{and} & t=0 \\
	\text{deg}=2: & X_+ = 0 & \text{and} & X_- \vee e_{2345}^\vee = 0 \\
	\text{deg}=4: & A(e_{2345}^\vee) = 0
	\end{matrix}
	\end{equation}
	A $4\times 4$-block of $A$ supplies the semisimple part $\mathrm{SL}(4)$, while the upper triangular $1 \times 4$-block combines with the four remaining components of $X_-$ to $N_8$.
\end{proof}

\numpar[][$\cP \cap \P(S_Q) = \{p\}$]
Here, we can always fix a decomposition so that $\psi_1 = \psi_p$ is a pure spinor, while $\psi_2=\psi_g$ lies in the generic open orbit. In other words, we are dealing with a subspace $S_Q$ where we can choose a basis so that one element is a pure spinor. In this decomposition, the square-zero condition reduces to
\begin{equation}
[Q,Q] = 2 \gamma(\psi_p, \psi_g) (w_1 , w_2)_W + \gamma(\psi_g , \psi_g) (w_2, w_2)_W = 0.
\end{equation}
Note that $\gamma|_{S_Q}$ defines a linear map $\Sym^2(S_Q) \longrightarrow V$. It makes sense to distinguish supercharges by the rank of this map. Note that $\Im(\gamma|_{S_Q}) = \mathrm{span}(\gamma(\psi_p, \psi_g) , \gamma(\psi_g, \psi_g))$ so that the rank can only be one or two. However, if the rank is two, the square-zero condition enforces $w_2$ to be isotropic and $w_1$ to be orthogonal to $w_2$. Since the line spanned by $w_2$ is maximal isotropic in $W$, this already implies that $w_1$ and $w_2$ are proportional so that we are back in the rank one case and do not find any additional square zero supercharges.

If $\rank(\gamma|_{S_Q}) = 1$, we have $\gamma(\psi_p, \psi_g) = \alpha \gamma(\psi_g, \psi_g)$ for some $\alpha \in \CC$. A short calculation shows that the condition $\P(S_Q) \cap \cP = \{p\}$ already implies that $\alpha = 0$. Indeed, for non-zero $\alpha$, the pair $(\psi_p , \psi_p - 2 \alpha \psi_g)$ is a basis of $S_Q$ where both elements are pure spinors.

Thus, we find that the square-zero conditions for such a supercharge $Q$ reduces to the requirement that $w_2$ is isotropic.

\begin{lem}
	$G$ acts transitively on supercharges of this type.
\end{lem}
\begin{proof}
	We show this, by bringing $Q = \psi_p \otimes w_1 + \psi_g \otimes w_2$ to a standard form via the group action.
	
	First, we can use the transitive action of $\Spin(V)$ on pure spinors to arrange for $\psi_p = 1 \in \wedge^0 L^\vee$. Then, we can write $\psi_g$ as $\psi_g = \psi_g^{(2)} + \psi_g^{(4)}$ with two- and four-form components. Since $\gamma(\psi_p, \psi_g) = 0$, we immediately find that $\psi_g^{(4)}=0$.
	Further, $\psi^{(2)}$ has to be of full rank for $\psi_g$ to be impure. Recall that the stabilizer of the line spanned by $\psi_p$ is $\mathrm{GL}(5) \ltimes \wedge^2 L$. We can use the $\mathrm{GL(5)}$ action to arrange for $\psi_g = e_{23}^\vee + e_{45}^\vee$. 
	
	Since $w_2$ is isotropic, we can apply $\mathrm{O}(2) \times \CC^\times$ to bring it to $(1,i)$. If $w_1$ is non-isotropic, we can immediately apply $\mathrm{O}(2)$ to bring it to $(a,0)$ for some $a \neq 0$ and then apply $\mathrm{diag}(a^{-1},1,1,1,1) \in \mathrm{GL}(5)$ so that we arrive at the standard form:
	\begin{equation} \label{eq: rep 0}
	1 \otimes (1,0) + (e_{23}^\vee + e_{45}^\vee) \otimes (1,i).
	\end{equation}
	If $w_1$ is isotropic, it is proportional to $(1,-i)$ and we can bring it to~\eqref{eq: rep 0} by first applying a contraction with $e_2 \wedge e_3 \in \wedge^2 L$ and then proceeding as above.
\end{proof}

\begin{lem}
	Supercharges of this type have a single surviving translation so that the corresponding flat background twist live on $\RR^8 \times \CC$. The stabilizer is $\mathrm{Sp}(4) \ltimes N_{13} \times \CC^\times$.
\end{lem}
\begin{proof}
	From the representatives above we find that $\Im([Q,-]) = \mathrm{span}(e_1, \dots , e_5, e^\vee_2, \dots, e^\vee_5)$. This identifies $e_1^\vee$ as the single surviving holomorphic direction in the twist.
	
	We compute the stabilizer using the representative from~\eqref{eq: rep 0}. Explicitly, we find the following conditions for the stabilizer.
	\begin{equation}
	\begin{matrix}
	\text{deg}=0: & \tr(A) - t= 0 & \text{and} & X_- \vee (e_{23}^\vee + e_{45}^\vee) + t = 0 \\
	\text{deg}=2: & X_+ = 0 & \text{and} & A(e_{23}^\vee + e_{45}^\vee) + it (e_{23}^\vee + e_{45}^\vee) =0 \\
	\text{deg}=4: & X_+ \wedge (e_{23}^\vee + e_{45}^\vee) = 0
	\end{matrix}
	\end{equation}
	We can investigate the condition posed on $A \in \mathrm{GL(5)}$ by introducing a block decomposition
	\begin{equation}
		A = \begin{pmatrix}
		\alpha & B \\
		C & D
		\end{pmatrix}
	\end{equation}
	where $\alpha \in \CC$, $D \in \mathrm{Mat}_{4 \times 4}(\CC)$, $B \in \mathrm{Mat}_{1 \times 4}(\CC)$ and $\mathrm{Mat}_{4 \times 1}(\CC)$. Viewing $e_{23}^\vee + e_{45}^\vee$ as a symplectic structure on $\mathrm{span}(e_2, \dots, e_5)$, we see that the degree 2 condition restricts $D$ to $\mathfrak{sp}(4) \times \CC$ and imposes $\mathrm{tr}(D) + t =0$. Further, it sets $B=0$ and leaves $C$ arbitrary. The degree zero condition further imposes $\alpha + \tr(D) = t$. A single linear equation on $X_-$ leaving $9$ free parameters.
\end{proof}

\numpar[][$\cP \cap \P(S_Q) =  \emptyset$]
Finally let us show that there are no square-zero supercharges of this type.
\begin{lem}
	There are no rank two square-zero supercharges with $\cP \cap \P(S_Q) =  \emptyset$.
\end{lem}
\begin{proof}
	Assuming we have such a $Q$, we can pick a basis of two isotropic vectors $(w_1,w_2)$ of $W$ and expand
	\begin{equation}
		Q = \psi_1 \otimes w_1 + \psi_2 \otimes w_2,
	\end{equation}
	where both $\psi_1$ and $\psi_2$ are impure spinors. The square-zero condition is then equivalent to $\gamma(\psi_1, \psi_2) =0$. We now show that this already implies $\gamma(\psi_1, \psi_1) = \alpha \gamma(\psi_2, \psi_2)$ for some $\alpha \in \CC$.
	
	Without loss of generality, we can use the $\Spin(10)$ action to arrange $\psi_1 = e_{23}^\vee + e_{45}^\vee$. In particular, we have $\gamma(\psi_1, \psi_1)= e_1$. Writing $\psi_2 = \psi_2^{(0)} + \psi_2^{(2)} + \psi_2^{(4)}$ we find that the square-zero condition gives two equations:
	\begin{equation}
		\psi_1 \vee \Omega^{-1}(\psi_2^{(4)}) = 0 \qquad \text{and} \qquad \psi_1 \wedge \psi_2^{(2)} = 0.
	\end{equation}
	Solving these, it is easy to see that $\gamma(\psi_2, \psi_2)$ is proportional to $e_1$.
	
	However, now we can expand a general element $\psi \in S_Q$ as a linear combination $\psi = a \psi_1 + b \psi_2$ to find
	\begin{equation}
	\gamma(\psi,\psi) = p(a,b) e_1
	\end{equation}
	for some symmetric quadratic polynomial $p$ in two variables. Every such polynomial has non-trivial zeros over $\CC$ so that we find that $S_Q$ necessarily contains a pure spinor which completes the proof.
\end{proof}

\subsection*{Declarations}
This research is supported by the DOE Early Career Research Program under award DE-SC0022924. The author has no competing interests to declare that are relevant to the content of this article.

\printbibliography

\end{document}